\documentclass[10pt,conference]{IEEEtran}

\usepackage{cite}
\usepackage{amsmath,amssymb,amsfonts}
\usepackage{algorithmic}
\usepackage{algorithm}
\usepackage{needspace}
\usepackage{booktabs}
\usepackage{makecell}
\usepackage{graphicx}
\usepackage{textcomp}
\usepackage{xcolor}
\usepackage{multirow}
\usepackage{booktabs}
\usepackage{lscape}
\usepackage{hologo}
\usepackage{amsthm}
\usepackage{enumitem}
\usepackage{geometry}
\usepackage{booktabs}
\usepackage[caption=false]{subfig}
\newtheorem{theorem}{Theorem}

\usepackage{mathtools}

\let\oldthebibliography\thebibliography
\let\endoldthebibliography\endthebibliography
\renewenvironment{thebibliography}[1]{%
  \small
  \oldthebibliography{#1}%
  \setlength{\itemsep}{0pt plus 0.3ex}%
}{\endoldthebibliography}

\usepackage{titlesec}
\titlespacing*{\section}{0pt}{8pt plus 2pt minus 2pt}{4pt plus 2pt minus 2pt}

\AtBeginDocument{
  \setlength\abovedisplayskip{4pt}
  \setlength\belowdisplayskip{4pt}
  \setlength\abovedisplayshortskip{2pt}
  \setlength\belowdisplayshortskip{2pt}
}

\setlist{nosep, leftmargin=*, topsep=2pt, partopsep=0pt}
\vspace{-0.9cm}
\begin{document}\title{\vspace{-0.9cm}Predictive Importance Sampling Based Coverage Verification for Multi-UAV Trajectory Planning}
% \begin{document}\title{Continuous Coverage Verification for Multi-UAV Trajectory Planning Using Predictive Importance Sampling}
\author{
\IEEEauthorblockN{Snehashish Ghosh and Sasthi C. Ghosh}
\IEEEauthorblockA{
Advanced Computing and Microelectronics Unit, Indian Statistical Institute, Kolkata, India\\
Email: cs2431@isical.ac.in, sasthi@isical.ac.in
}
}
% Note: 
% 1) Add SAC reference 2)make the result texts enlarged 3) HO plot needs to be redesigned. 4) UAV predicts user NLoS in upcoming $\Delta t$ time
% 5) give reference after snapshot in introduction. 6)more referrences in related work section. 7)Past info motivation before 'plenty of research'
% 7) Motivate intra timestep, but multi UAV(in related work -abid, siddhanta) and their sampling method is in-efficient.
% 8)use output of algo 1 to algo 2 9) show multi-UAV coordination, algo 2 is for 1 uav only, how different UAVs are coordinated(how actor critic works to ensure) 10) how R(t) is used for multi-uav coordination, how computation of centroid is used there

% 11) Need to recheck the computational complexity
% 12) Algo 2 complexity.
% 13) write explicitly, how continuous movement => rare event Monte Carlo simulation
% \\
% \\
% Doubt: 
% 1) For uniform distribution should we write p(z) or simply U?
% 2) normalized reward-penalty sign should be '+' or '-' 3) what should be the value of $\Delta_t$?

\maketitle

\begin{abstract}
Unmanned aerial vehicle (UAV) networks are emerging as a promising solution for ultra-reliable low-latency communication (URLLC) in next-generation wireless systems. A key challenge in millimeter wave UAV networks is maintaining continuous line of sight (LoS) coverage for mobile users, as existing snapshot-based trajectory planning methods fail to account for user mobility within decision intervals, leading to catastrophic coverage gaps. Standard uniform sampling for continuous coverage verification is computationally prohibitive, requiring huge number of samples to estimate rare failure events with latencies incompatible with real-time requirements. In this work, we propose a predictive importance sampling (PIS) framework that drastically reduces sample complexity by concentrating verification efforts on predicted failure regions. Specifically, we develop a long short-term memory mixture density network (LSTM-MDN) architecture to capture multimodal user trajectory distributions and combine it with defensive mixture sampling for robustness against prediction errors. We prove that PIS provides unbiased failure probability estimates with lower variance than uniform sampling. We then integrate PIS with multi-agent deep deterministic policy gradient (MADDPG) for coordinated multi-UAV trajectory planning using an adaptive multi-objective reward function balancing throughput, coverage, fairness, and energy consumption. Lastly, the simulation results show how our suggested method outperforms three other state-of-the-art methods in terms of coverage rate, throughput, and verification latency, making proactive coverage management for URLLC-aware UAV networks feasible.
\end{abstract}


\section{Introduction}
\label{sec:intro}

\noindent Unmanned aerial vehicle (UAV) networks have become a versatile way to handle ultra-reliable low-latency communication (URLLC) needs, such as emergency  services, industrial IoT, and autonomous vehicles alongside enhanced mobile broadband (eMBB) services for high-throughput data-intensive applications \cite{mozaffari2019tutorial,saad2020vision}. 
While eMBB users can tolerate minor delays, URLLC users require sub-millisecond latency and very high reliability, which means the network must be managed proactively instead of just reacting to the problems. Over short distances, millimeter-wave (mmWave) can provide high data speeds and dependability.
However, due to substantial propagation and penetration losses of mmWave, line-of-sight (LoS) conditions are the main factor in link quality, so maintaining a steady LoS connection is vital for users on the move~\cite{al2014optimal}.
% The majority of current UAV trajectory planning techniques rely on snapshots, meaning they only verify coverage at particular time points~\cite{Jin2024Trajectory}. This assumes that users stay stationary during the decision time interval($\Delta t$) between two consecutive time points~\cite{zeng2017energy,chowdhury2020mobility}. However, in the real world, users are often mobile even during the $\Delta t$ interval. A user who has a LoS signal at time $t$ might have intermediary non-LoS (NLoS) due to the presence of small obstacles like portions of buildings (say, a building edge) while moving to their next time point $t+\Delta t$ . 
Plenty of research focuses on where to put UAVs~\cite{alzenad2017d} and how to plan their paths~\cite{zeng2017energy} to maximize coverage. Recent studies use deep reinforcement learning like multi-agent deep deterministic policy gradient (MADDPG) for coordinating multiple UAVs trajectory planning~\cite{VanillaMADDPG2025, ravi2025optimizing}. 
 While these methods achieve effective trajectory planning, their coverage verification mechanisms have critical limitations. Qin et al. \cite{noma_2025} verify coverage through SINR thresholds and queue-based URLLC constraints, but only at discrete time slots under quasi-static assumptions. Dhinesh Kumar et al. \cite{MADDPG_BO2025} use binary coverage indicators that check if data rates exceed thresholds at specific checkpoints, achieving good snapshot-based verification. However, all these approaches share a fundamental limitation, they only verify coverage at discrete time points and ignore what happens during user movement within decision intervals $\Delta t$. This creates a critical gap where users with LoS at times $t$ and $t+\Delta t$ may still experience intermediate non-LoS (NLoS) due to obstacles during movement. These methods lack predictive mechanisms to anticipate where failures are likely to occur, leading to inefficient allocation of computational resources.
More precisely, a user's LoS connectivity with the UAV at the start and finish of that time interval does not always imply LoS connectivity throughout the duration of that period. (Fig.~\ref{fig:coverage_gap}). This leads to dropped connections, which can be devastating for delay-sensitive URLLC based mission critical applications as well as makes a system less reliable for eMBB users. It can be addressed by considering users’ location at very short intervals by making $\Delta t$ very small. However, in that case, location updates will take up a significant portion of data communication time, thereby reducing the user throughput substantially.
% Adding to the draw
\begin{figure}
    \centering
    \includegraphics[width=0.8\linewidth]{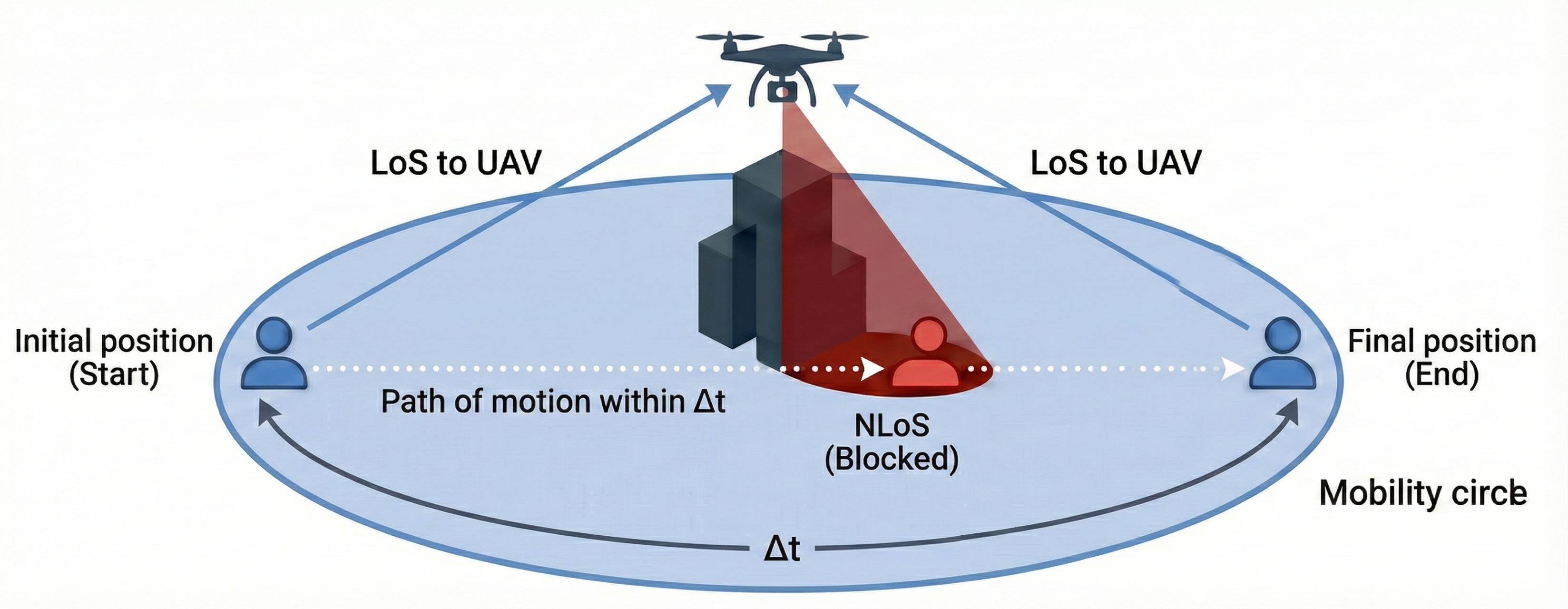}
    \caption{Analysis of coverage gaps during user movement.}
    \label{fig:coverage_gap}
    % \vspace{-0.6cm}
\end{figure}
% \begin{figure}
%     \centering
%     \includegraphics[width=1\linewidth]{coverage_gap (2).pdf}
%     \caption{Analysis of coverage gaps during user movement.}
%     \label{fig:coverage_gap}
% \end{figure}

If no past mobility related information is considered, at a particular time point $t$, a user $u$ positioned at $u_t$ and moving at velocity $v$ can end up anywhere inside the mobility circle of radius $v \cdot \Delta t$ with center as $u_t$. To ensure complete coverage, we should confirm that no point within $C$ falls into NLoS. Checking every point via ray-casting is impossible. Instead, we have to use probabilistic estimation to find the failure probability $P_f$. For high reliability, $P_f$ must be very small. Using standard uniform Monte Carlo estimation, $N = O(1/\sqrt{P_f})$ samples are required~\cite{rubinstein2017simulation} implying approximately $10^6$ samples for $P_f = 10^{-4}$. Assuming that each sample takes about 1$\mu$s for ray-casting, this corresponds to a latency of roughly $1 \; to \;100$ ms per user. This is too slow for a real-world system that requires sub-millisecond latency.

One way to reduce this latency is to consider the past mobility related information of the users. Long short-term memory (LSTM) networks are great for modeling this type of past information sequences~\cite{hochreiter1997long}. However, predicting a single deterministic trajectory introduces significant uncertainty as the users at intersections can take multiple possible paths, and averaging these possibilities often yields physically unrealistic predictions. 
To address this, mixture density networks (MDNs)~\cite{bishop1994mixture} can be integrated with LSTMs to model the inherent multimodality of user movement, predicting multiple possible trajectories with associated probabilities. Moreover, by sampling preferentially in these high-probability directions rather than uniformly throughout the entire mobility region $C$, more efficient coverage verification can be achieved. Such non uniform sampling can efficiently be modeled by using classical {\it importance sampling} (IS).

To the best of our knowledge, no existing research covers (1) efficient coverage checking for continuous movement between timesteps, (2) combining learned path prediction with formal IS theory, or (3) real-time coordination for multiple UAVs with lower variance.
Building on these insights, we develop a predictive importance sampling (PIS) which provably reduces variance in Monte Carlo estimation~\cite{owen2013monte,rubinstein2017simulation} and consequently reduces verification latency by focusing computational resources on the most critical regions. Using this, we develop an efficient multi-UAV trajectory planning mechanism.
More specifically our contributions can be summarized as follows:
\begin{enumerate}
\item We develop a theoretically-grounded PIS framework for efficient coverage verification, and prove that it provides unbiased failure probability estimates with reduced variance versus uniform sampling under specified conditions. Unbiasedness ensures accurate URLLC reliability guarantees, while variance reduction enables real-time verification with significantly fewer samples.

\item We create an LSTM-MDN architecture as the predictive foundation that generates multi-modal trajectory distributions focusing sampling on anticipated failure locations, employing a defensive mixture approach combining prediction with uniform sampling backup for robustness.

\item We develop an adaptive multi-objective reward function with learnable weights that dynamically balances throughput, coverage, fairness, and energy efficiency based on network state.

\item We integrate PIS with MADDPG to design coordinated multi-UAV trajectory planning enabling real-time proactive coverage management, and establish our strategy's superiority over state-of-the-art approaches through comprehensive experiments.
\end{enumerate}

%However, these techniques haven't been adapted for real-time wireless checks. \\

%Mixture Density Networks (MDNs)~\cite{bishop1994mixture} are good at capturing complex data shapes. While these have been used for autonomous systems~\cite{alahi2016social,lee2017desire}, they haven't been used as the foundation for importance sampling in trajectory planning.

% \subsection{Our contributions are listed as follows:}

The rest of this paper is organized as follows: Section \ref{sec:system} describes the system model, Section \ref{prop_str} presents the proposed PIS framework along with its proofs of unbiasedness and variance reduction, LSTM-MDN architecture, and coordinated multi-UAV trajectory planning, Section \ref{simul_res} provides simulation results and comparison with existing related strategies, and Section \ref{conclusion} concludes with possible future directions.

\section{System Model}
\label{sec:system}

\subsection{Network Topology}

\noindent We consider a three-dimensional air-to-ground cellular wireless communication system consisting of $N_{\text{UAV}}$ UAVs serving as aerial base stations and $N_{\text{users}}$ ground users over a geographical area of $a \times a$ m$^2$. 
Here, both UAVs and users are mobile. The environment contains $M$ static obstacles (buildings) that can block LoS paths between UAVs and ground users. Each UAV operates within an altitude range of $h_{min}$ to $h_{max}$ meters with maximum speed $v_{\text{UAV}}$ m/s to avoid collision with buildings as well to ensure {\it safe flight zone} within the designated green zone. We assume ground users are located at zero altitude. Time is discretized into decision intervals of duration $\Delta t$ seconds.

%We consider a three-dimensional air-to-ground cellular wireless communication system consisting of $N_{\text{UAV}} = 5$ UAVs serving as aerial base stations and $N_{\text{users}} = 230$ ground users over a geographical area of $1.5 \times 1.5$ km$^2$. The environment contains $M = 50$ static obstacles (buildings) with heights ranging from 20 to 80 meters that can block Line-of-Sight (LoS) paths between UAVs and ground users. The user population comprises $N_{\text{URLLC}} = 30$ Ultra-Reliable Low-Latency Communication (URLLC) users and $N_{\text{eMBB}} = 200$ enhanced Mobile Broadband (eMBB) users. Each UAV is equipped with a single antenna and operates within an altitude range of 100-500 meters with maximum speed $v_{\text{UAV}} = 25$ m/s. To ensure flight safety, UAVs try to maintain a separation distance of $d_{\text{min}} = 350$ meters. Time is discretized into decision intervals of duration $\Delta t = 300$ second.

% \subsection{Channel Model- \textbf{requires complete overhaul}} 

\subsection{Channel Model}
\noindent We adopt the 73 GHz air-to-ground propagation model with both LoS and NLoS conditions. The path loss (in dB) for the link between UAV $i$ and user $j$ is~\cite{Liu2014Capacity}:
\begin{equation}
    PL_{i,j}(d_{i,j}) = \alpha_{pl} + 10\beta_{pl}\log_{10}(d_{i,j}),
\end{equation}
where $\alpha_{pl}$ and $\beta_{pl}$ are environment-specific constants and $d_{i,j}$ is the Euclidean distance. The received power (in dBm) at user $j$ from UAV $i$ is:
\begin{equation}
    P_{i,j} = P_{i} + G_{i} + G_{j} - PL_{i,j}(d_{i,j}),
\end{equation}
where $P_{i}$ is the transmit power and $G_{i}$, $G_{j}$ are antenna gains at the UAV and user, respectively. The achievable throughput is given by Shannon's capacity~\cite{Liu2014Capacity}:
\begin{equation}
    T_{i,j} = B_{w}\log_{2}\left(1+\frac{P_{i,j} \times |g_{i,j}|^{2}}{\sigma_{0}^{2}}\right),
\end{equation}
where $B_{w}$ is the channel bandwidth, $\sigma_{0}^{2}$ is the additive white Gaussian noise power, and $|g_{i,j}|$ is the small-scale fading coefficient following Rician (LoS) or Rayleigh (NLoS) distributions.

%Note that, here, we assume orthogonal frequency division multiple access (OFDMA) where each UE communicates using orthogonal resource blocks [].

%Due to the high path loss at mmWave frequencies, we consider only LoS links for successful communication. 
% In our system, UAVs and users are mobile entities. Due to the presence of static obstacles, the direct LoS path between a UAV and ground user may be blocked. We consider only LoS links for successful communication at 61 GHz carrier frequency. For a particular UAV $u_i$ at position $\mathbf{p}_a \in \mathbb{R}^3$ and ground user $u_j$ at position $z \in \mathbb{R}^2$, the channel gain is modeled using a simplified path loss model:
% \begin{equation}
% g(d) = \frac{1}{d^{2.5}},
% \end{equation}
% where $d = \sqrt{(x_a - x_j)^2 + (y_a - y_j)^2 + z_a^2}$ is the three-dimensional Euclidean distance in meters. The received Signal-to-Noise Ratio (SNR) at user $u_j$ is:
% \begin{equation}
% \text{SNR}_{i,j} = \frac{P_{\text{tx}} \cdot g(d)}{N_0},
% \end{equation}
% where $P_{\text{tx}} = 250$ mW (24 dBm) is the transmit power and $N_0 = 10^{-13}$ W is the noise power. The achievable data rate follows Shannon capacity:
% \begin{equation}
% C_{i,j} = B \log_2(1 + \text{SNR}_{i,j}),
% \end{equation}
% where $B = 1$ MHz is the channel bandwidth. The communication range of each UAV is $R_{\text{comm}} = 300$ meters. LoS determination between UAV and ground point is performed via computational geometry-based ray-casting, checking whether the line segment intersects any obstacle.

\subsection{User Traffic Characterization}

\noindent According to user delay tolerance requirements, we classify users into two service classes namely URLLC and eMBB. For simplicity we have not considered massive machine-type communication (mMTC) users as we focus only on human type applications. URLLC traffic is designed for mission-critical applications requiring sub-millisecond latency and high reliability. Even minor coverage interruptions during user movement can cause catastrophic failures. URLLC users require continuous LoS coverage verification throughout their mobility region during decision interval $\Delta t$. A user is considered covered if the estimated failure probability $\hat{P}_f$ below a pre-defined threshold.
eMBB traffic has been designed for high-throughput, data-intensive applications that can tolerate minor delays. eMBB users do not require  extensive LoS coverage verification.

\subsection{Mobility Model} 
\noindent User velocities are bounded by $v_u \in [0, v_{max}]$ m/s, representing pedestrian and vehicular mobility. Users remain within the simulation area through boundary reflection. At a particular time point $ t \in [0,\Delta t]$, a user at position $u_t$ with velocity $v_u$ can occupy any position within the mobility circle:
\begin{equation}
C = \{z \in \mathbb{R}^2 : \|z - u_t\| \leq r\}.
\end{equation}
where $r = v_u \cdot \Delta t$.
UAVs execute continuous control actions $\Delta p = (\Delta x, \Delta y, \Delta z)$ generated by actor networks, constrained by the velocity $v_{uav} \in [0, V_{\text{max}}]$ m/s. Position updates are:
\begin{equation}
p_{t+1} = p_t + v_{uav} \cdot \Delta t,
\end{equation}
with hard boundary constraints ensuring UAVs remain within the safe flight zone.

\section{Proposed Strategy}
\label{prop_str}
% \noindent This section presents our predictive importance sampling framework for efficient user coverage verification in UAV networks. We begin by formulating the rare event challenge that makes standard uniform Monte Carlo estimation computationally intractable for ultra-low failure probabilities (\ref{3.1}). We then introduce importance sampling principles and derive the optimal proposal distribution (\ref{3.2}), followed by our practical PIS model that combines LSTM-based trajectory prediction with defensive mixture sampling (MDN) (\ref{3.3}). We prove the unbiasedness and variance reduction properties of our estimator (\ref{3.4}), describe the LSTM-MDN architecture for multimodal trajectory prediction (\ref{3.5}). Finally we present the MADDPG-based coordinated multi-UAV planning with adaptive multi-objective rewards (\ref{3.6}). 

\noindent This section presents our predictive importance sampling framework for efficient user coverage verification in UAV networks. We formulate the rare event challenge (\ref{3.1}), introduce importance sampling principles and derive the optimal proposal distribution (\ref{3.2}), then present our practical PIS model combining LSTM-based trajectory prediction with defensive mixture sampling using MDN (\ref{3.3}). We prove the unbiasedness and variance reduction properties of our estimator (\ref{3.4}), describe the LSTM-MDN architecture for multimodal trajectory prediction (\ref{3.5}), and present MADDPG-based coordinated multi-UAV planning with adaptive multi-objective rewards (\ref{3.6}).

\subsection{Rare Event Challenge}
\label{3.1}
\noindent As mentioned earlier unlike eMBB, URLLC requirements mandate continuous coverage throughout $C$. Formally, we must verify $L(z) = 1$ for all $z \in C$, where $L(z)$ denotes the LoS status of the point $z$. Since verifying uncountably infinite points is intractable, we estimate failure probability. Assuming worst-case uniform uncertainty over $C$, the failure probability is:
\begin{equation}
P_f = \frac{\int_C I_F(z) \, dz}{\int_C dz} = \mathbb{E}_p[I_F(Z)],
\end{equation}
where $Z \sim \text{Uniform}(C)$,  $p(z) = 1/A$ with $A$ being the area of the mobility circle $C$ and $I_F(z)=1(0)$ is the indicator variable showing the NLoS (LoS) at position $z$. Standard Monte Carlo estimator is
\begin{equation}
\hat{P}_f^{\text{uniform}} = \frac{1}{N} \sum_{i=1}^{N} I_F(Z_i), \quad Z_i \sim \text{Uniform}(C).
\end{equation}
Now we come across the \textit{rare event challenge} where probability of NLoS is very rare which may require extraordinary effort to locate the NLoS. To address this rare event challenge we adopt the importance sampling principle.
% If variance is $\text{Var}(\hat{P}_f) \approx P_f/N$ then relative error (RE) is
% \begin{equation}
% \label{rel_err}
% \text{RE} = \frac{\sqrt{\text{Var}(\hat{P}_f)}}{P_f} \approx \frac{1}{\sqrt{N \cdot P_f}}.
% \end{equation}

% This implies that we need $N = O(1/\sqrt{P_f})$ samples to get accurate results~\cite{rubinstein2017simulation}. For $P_f = 10^{-4}$, it corresponds to $10^6$ ($10^8$) samples to attain a 10\% (1\%) relative error. At an estimated cost of 1$\mu$s per ray-casting sample, this results in a latency of about 1 ms (100 ms). 
% Such latency is impractical for a real world system that requires a sub-millisecond latency.
% This implies that for $P_f = 10^{-4}$ and target RE = 10\%, we need $N \geq 10^6$ samples, which is computationally infeasible for real-time control.
%To address this rare event challenge we adopt the importance sampling.

% \section{Predictive Importance Sampling Framework}

\subsection{Importance Sampling (IS) Principles}
\label{3.2}
\noindent IS is a variance reduction technique that changes the sampling distribution while maintaining unbiasedness through importance weights~\cite{owen2013monte}. The fundamental IS identity states that for any expectation $\mu = \mathbb{E}_p[f(X)]$ under distribution $p(x)$:
\begin{equation}
\mu = \int f(x) p(x) \, dx = \mathbb{E}_q\left[f(X) \frac{p(X)}{q(X)}\right]
\end{equation}
where $q(x)$ is any proposal distribution with support containing that of $p(x)$. The term $w(x) = p(x)/q(x)$ is the importance weight. The IS estimator is \cite{rubinstein2017simulation}:
\begin{equation}
\hat{\mu}_q = \frac{1}{N} \sum_{i=1}^{N} f(X_i) w(X_i), \quad X_i \stackrel{\text{i.i.d.}}{\sim} q(x).
\end{equation}

This estimator is \textit{unbiased} meaning its expected value is equal to the true mean $\mu$, provided the support condition $q(x) > 0$ wherever $p(x) > 0$ is met~\cite{rubinstein2017simulation}. The proof is a direct consequence of the IS identity:
\begin{equation}
    \mathbb{E}_q[\hat{\mu}_q] =  \mathbb{E}_q\left[ f(X) w(X)\right] = \mu.
\end{equation}

Its variance is:
\begin{equation}
\label{variance}
\text{Var}_q(\hat{\mu}_q) = \frac{1}{N}\left[\mathbb{E}_q\left[(f(X)w(X))^2\right] - \mu^2\right].
\end{equation}
% The above identity is based on the condition $f(x) \geq 0$ for all $x$.
If $f(x) \geq 0$ for all $x$, the \textit{optimal proposal distribution} that minimizes variance is:

\begin{equation}
\label{optimal_q}
q^*(x) = \frac{f(x)p(x)}{\int f(y)p(y) \, dy}  = \frac{f(x)p(x)}{\mu}\mbox{~\cite{rubinstein2017simulation}}.
\end{equation}

Based on IS principles, we now propose an intelligent predictive importance sampling (PIS) model for coverage verification.

\subsection{Proposed PIS model for Coverage Verification}
\label{3.3}
\noindent For our coverage verification problem, we have  $f(x) = I_F(x)$. Since $I_F(x) \in \{0, 1\}$, the condition $f(x) \geq 0$ for the optimal proposal distribution in \eqref{optimal_q} is satisfied. If $f(x) = I_F(x)$, $f(x)w(x) = f(x)\frac{p(x)}{q*(x)} = \mu$ implying that variance as stated in \eqref{variance} becomes zero. However, as we can see from \eqref{optimal_q}, constructing $q^*$ requires knowing $\mu$ itself, the very quantity we seek to estimate. Therefore, practical IS designs $q(x)$ to approximate $q^*(x)$. In our coverage verification problem, $f = I_F$ and $p$ is uniform implying $p(x) = 1/A$. Since $\mu$ is constant, from \eqref{variance}, we get $q^*(z) \propto I_F(z)$. This implies that  the optimal proposal concentrates samples on NLoS regions.
Note that our objective is to design the learned $q(z)$ to approximate $q^*(z)$. A critical failure will occur when $q(z) \approx 0$ where $p(z) > 0$. This causes importance weights $w(z) = p(z)/q(z) \rightarrow \infty$, implying a catastrophic increase in variance. This is particularly dangerous with learned proposals that may mis-predict failure (NLoS) regions.
To overcome this, we employ a defensive mixture strategy:
% \vspace{-0.05cm}
\begin{equation}
q_{\text{pis}}(z) = \alpha \cdot q_{\text{pred}}(z) + (1-\alpha) \cdot p(z), \quad \alpha \in [0,1)
\end{equation}
where $q_{\text{pred}}(z)$ is the LSTM-MDN prediction concentrating on predicted failure regions and $p(z)$ is the uniform distribution ensuring {\it fallback coverage}. 
Since $p(z) = 1/A > 0$ everywhere in $C$, we have $q_{\text{pis}}(z) \geq (1-\alpha)/A > 0$ for all $z \in C$, guaranteeing bounded importance weights $w(z) \leq A/(1-\alpha)$.
Here the parameter $\alpha$ provides tunable control: higher the $\alpha$ (tends to 1) more is the reliance on prediction whereas low $\alpha$ (tends to 0) prioritizes robust uniform coverage. So, $\alpha$ must be chosen empirically to balance both the objectives.

\subsection{Proof of Unbiasness and Variance Reduction of PIS}
\label{3.4}
\begin{theorem}
\label{thm:unbiased}
The PIS estimator 
\begin{equation}
\hat{P}_{\text{PIS}} = \frac{1}{N}\sum_{i=1}^{N} I_F(Z_i) \frac{p(Z_i)}{q_{\text{pis}}(Z_i)}, \quad Z_i \stackrel{\text{i.i.d.}}\sim q_{\text{pis}}
\end{equation}
is unbiased: $\mathbb{E}_{q_{\text{pis}}}[\hat{P}_{\text{PIS}}] = P_f$.
\end{theorem}

\begin{proof}
By linearity of expectation and i.i.d. samples:
\begin{align}
\mathbb{E}_{q_{\text{pis}}}[\hat{P}_{\text{PIS}}] &= \mathbb{E}_{q_{\text{pis}}}\left[I_F(Z) \frac{p(Z)}{q_{\text{pis}}(Z)}\right] \\
&= \int_C I_F(z) \frac{p(z)}{q_{\text{pis}}(z)} q_{\text{pis}}(z) \, dz \\
&= \int_C I_F(z) p(z) \, dz = \mathbb{E}_p[I_F(Z)] = P_f
\end{align}
where $q_{\text{pis}}(z)$ terms cancel since $q_{\text{pis}}(z) > 0 \; \forall z$.
\end{proof}

\begin{theorem}
\label{thm:variance}
Let $\mathcal{F} = \{z \in C : I_F(z) = 1\}$ be the failure region. If $q_{\text{pis}}(z) \geq p(z)$ for all $z \in \mathcal{F}$, then:
\begin{equation}
\text{Var}_{q_{\text{pis}}}(\hat{P}_{\text{PIS}}) \leq \text{Var}_{p}(\hat{P}_{\text{uniform}})
\end{equation}
\end{theorem}

\begin{proof}[Proof:]
The variance of the uniform estimator is:
\begin{align}
\label{uniform_var}
\text{Var}_p(\hat{P}_{\text{uniform}}) &= \frac{1}{N}[\mathbb{E}_p[I_F(Z)^2] - P_f^2] \nonumber \\
&= \frac{1}{N}[P_f - P_f^2] = \frac{P_f(1-P_f)}{N},
\end{align}
since $I_F(Z)^2 = I_F(Z)$ for binary indicators.\\ The variance of the PIS estimator is:
\begin{equation}
\label{pis_var}
\text{Var}_{q_{\text{pis}}}(\hat{P}_{\text{PIS}}) = \frac{1}{N}\left[\mathbb{E}_{q_{\text{pis}}}[(I_F(Z)w(Z))^2] - P_f^2\right].
\end{equation}

For variance reduction, we require $\text{Var}_{q_{\text{pis}}} \leq \text{Var}_p(\hat{P}_{\text{uniform}})$.
Therefore, after equating \eqref{uniform_var} \& \eqref{pis_var} we get:
\begin{equation}
\mathbb{E}_{q_{\text{pis}}}[(I_F(Z)w(Z))^2] \leq P_f.
\end{equation}

Converting to expectation over $p$:
\begin{align}
\mathbb{E}_{q_{\text{pis}}}[(I_F w)^2] &= \int_C I_F(z)^2 w(z)^2 q_{\text{pis}}(z) \, dz \nonumber \\
&= \int_C I_F(z) \frac{p(z)^2}{q_{\text{pis}}(z)} \, dz \nonumber \\
&= \mathbb{E}_p\left[I_F(Z) \frac{p(Z)}{q_{\text{pis}}(Z)}\right].
\end{align}

% So, the fundamental condition for variance reduction is: 
% \begin{equation}
% \label{var_red}
%     \mathbb{E}_p\left[\frac{I_F(Z) p(Z)}{q_{\text{pis}}(Z)}\right] \leq \mathbb{E}_p[I_F(Z)] = P_f.
% \end{equation}

% The above inequality holds if $\frac{p(Z)}{q_{\text{pis}}(Z)} \leq 1 \implies q_{\text{pis}}(z) \geq p(z)$ for all $z \in \mathcal{F}$. If $q_{pis}(z) \geq p(z)$ then the ratio $\frac{p(z)}{q_{pis}(z)} \leq 1$ for all $z \in F$. Since $I_F(z) \in \{0,1\}$, we have $I_F(z) \frac{p(z)}{q_{pis}(z)} \leq I_F(z)$ for all $z \in C$. Taking the expectation with respect to $p(z)$ on both sides preserves the inequality:
% \[
% E_p \left[ I_F(Z) \frac{p(Z)}{q_{pis}(Z)} \right] \leq E_p [I_F(Z)] = P_f
% \]
% So, 
% completing the proof.
% \end{proof}
% \needspace{4\baselineskip} % Checks if there are at least 4 lines of space left
So, the variance reduction holds if and only if: 
\begin{equation}
\label{var_red}
    \mathbb{E}_p\left[\frac{I_F(Z) p(Z)}{q_{\text{pis}}(Z)}\right] \leq \mathbb{E}_p[I_F(Z)] = P_f.
\end{equation}

% \[
% \implies \frac{I_F(Z) p(Z)}{q_{\text{pis}}(Z)} \leq I_F(Z)
% \]

% \[
% \implies \frac{ p(Z)}{q_{\text{pis}}(Z)} \leq 1
% \]
% We are able to cancel out $I_F(Z)$ since it is $>= 0$.
% \[
% \implies  p(Z) \leq q_{\text{pis}}(Z)
% \]
% \end{proof}

We now show that the pointwise condition $q_{\text{pis}}(z) \geq p(z)$ for all $z \in \mathcal{F}$ is sufficient to guarantee \eqref{var_red}.

Now, if $q_{\text{pis}}(z) \geq p(z)$ holds for all $z \in \mathcal{F}$, then we get $\frac{p(z)}{q_{\text{pis}}(z)} \leq 1$. Since $I_F(z) = 0 \;\text{for all } z \notin \mathcal{F}$, we have:
\begin{equation}
I_F(z) \frac{p(z)}{q_{\text{pis}}(z)} \leq I_F(z) \quad \text{for all } z \in C
\end{equation}

Taking expectations with respect to $p(z)$:
\begin{equation}
\mathbb{E}_p\left[I_F(Z) \frac{p(Z)}{q_{\text{pis}}(Z)}\right] \leq \mathbb{E}_p[I_F(Z)] = P_f
\end{equation}

This establishes \eqref{var_red}, completing the proof.
\end{proof}

Theorems \ref{thm:unbiased} and \ref{thm:variance} provide a clear design principle that the predictive model $q_{\text{pred}}$ must concentrate sufficient probability mass on failure regions such that $q_{\text{pis}}(z) \geq p(z) \;\text{for all } z \in C$ holds. If the model $q_{\text{pred}}$ accurately identifies NLoS areas, this condition is satisfied, yielding variance reduction. If it cannot do so, the variance may exceed uniform sampling, however, it remains \textit{bounded} and controllable via $\alpha$.

\subsection{Trajectory Prediction using LSTM-MDN Architecture}
\label{3.5}
\noindent To determine the optimal proposal distribution $q^* \propto I_F(z)$, we must accurately predict user locations and NLoS regions by modeling their historical movement. User movement exhibits strong trends. At time $t$, we observe $H$ sequences of historical velocity data $\{(\mathbf{v}_{x,t-H}, \mathbf{v}_{y,t-H}), \ldots, (\mathbf{v}_{x,t-1}, \mathbf{v}_{y,t-1})\}$ and feed them into a 2-layer LSTM with hidden dimension 128, capturing movement patterns in the final hidden state $h_t$.
Standard networks predict a single trajectory, but real-world user movement is multimodal with multiple possible directions. For example, a user at a street corner can turn left, right, or continue straight, averaging these yields physically unrealistic predictions (e.g., walking into a building). MDNs~\cite{bishop1994mixture} solve this by outputting Gaussian mixture model (GMM) parameters instead of a single prediction. The MDN computes $K$ possible paths using LSTM state $h_t$: mixture weights $\pi_k(h_t)$ (summing to 1), mean positions $\boldsymbol{\mu}_k(h_t) \in \mathbb{R}^2$, and covariances $\boldsymbol{\Sigma}_k(h_t) \in \mathbb{R}^{2 \times 2}$. The predictive distribution is:
\begin{equation}
q_{\text{pred}}(z | h_t) = \sum_{k=1}^{K} \pi_k(h_t) \mathcal{N}(z | \boldsymbol{\mu}_k(h_t), \boldsymbol{\Sigma}_k(h_t)).
\end{equation}
The model is trained using negative log-likelihood (NLL) loss. The architecture comprises an input layer (velocity sequences of size $H \times 2$), a 2-layer LSTM, and an MDN output head with softmax functions on weights $\pi_k$, linear functions on means $\boldsymbol{\mu}_k$, and exponential functions ensuring positive diagonal covariance.

We now present \textbf{Algorithm~\ref{alg:pis1}} which computes PIS failure probability. It begins by processing the user's velocity history $\mathcal{H}_v$ via the LSTM-MDN to generate trajectory predictions. A defensive mixture proposal $q_{pis}(z)$ is constructed to ensure sampling robustness. The core loop executes IS: it draws samples from $q_{pis}$, verifies LoS, and computes importance weights. The aggregated failure probability $\hat{P}_f$ is then thresholded against $\epsilon$ to output a binary coverage quality metric $Q_u$. This $Q_u$ serves as the statistically verified ground truth for the reward function in the next stage. Complexity of Algorithm~\ref{alg:pis1} is $O(N \cdot C_{los}) \approx O(N)$, where $N$ is the sample count and $C_{los}$ is the cost of LoS check (typically $O(1)$ with height maps).

\begin{algorithm}[h!]
\caption{PIS Failure Probability Estimation}
\label{alg:pis1}
\begin{algorithmic}[1]
\REQUIRE User state $u_t$, Velocity history $\mathcal{H}_v$, Sample count $N$, mixture weight $\alpha$.
% UAV candidate pos $\mathbf{p}_{next}$, Map $\mathcal{M}$, 
\ENSURE Binary coverage Quality $Q_u$.
\STATE \textbf{Predict:} Get parameters $\{\pi_k, \mu_k, \Sigma_k\} \leftarrow \text{LSTM-MDN}(\mathcal{H}_v)$
\STATE \textbf{Construct Proposal:} $q_{pis}(z) = \alpha \sum \pi_k \mathcal{N}(z|\mu_k, \Sigma_k) + (1-\alpha)U(\mathcal{C})$ \COMMENT {Where $U(\mathcal{C})$ is Uniform PDF: $1/\text{Area}(C)$}
\STATE Initialize weighted failure sum $W_{fail} \leftarrow 0$
\FOR{$i = 1$ to $N$}
    \STATE Sample potential user location $z_i \sim q_{pis}(z)$
    \STATE Reject if $\|z_i - u_t\| > \text{Radius(C)}$ (ensure $z_i \in C$)
    \STATE Check Line-of-Sight: $I_F(z_i)$
    
    % \leftarrow \text{RayCast}(\mathbf{p}_{next}, z_i, \mathcal{M})$
    \IF{$I_F(z_i) == 1$}
        \STATE Compute Importance Weight: $w_i \leftarrow \frac{1/\text{Area}(\mathcal{C})}{q_{pis}(z_i)}$
        \STATE $W_{fail} \leftarrow W_{fail} + w_i$
    \ENDIF
\ENDFOR
\STATE Estimate Failure Prob: $\hat{P}_f \leftarrow W_{fail} / N$
\STATE Determine Quality: \textbf{if} $\hat{P}_f < \epsilon$ \textbf{then} $Q_u \leftarrow 1$(Success) \textbf{else} $Q_u \leftarrow 0$(Failure)
\RETURN $Q_u$
\end{algorithmic}
\end{algorithm}

Until now we assumed single UAV for user trajectory prediction through coverage failure estimation. In order to avoid collisions and a significant overlap of served users, we suggest using multiple UAVs with coordinated operations. In the next subsection, we use multiple UAVs as agents and employ MADDPG~\cite{VanillaMADDPG2025} to coordinate the operations.
\vspace{-0.2cm}
\subsection{Coordinated Multi-UAV Trajectory Planning}
% \vspace{-0.2cm}
\label{3.6}
\noindent MADDPG uses centralized training with decentralized execution.
We have designed a multi-objective reward function for proactive coverage maintenance:
% \begin{equation}
% R_t = R_{\text{cov}} + R_{\text{throughput}} + R_{\text{density}} + R_{\text{balance}}  - P_{\text{energy}} - P_{\text{collision}}
% \end{equation}
\begin{equation}
R_t = \sum_{j\in R} \omega_{j} \cdot \Phi_j(R_j) + \sum_{k\in P} \omega_{k} \cdot \Phi_k(P_k),
\end{equation}

where $\Phi(\cdot)$ is  a normalization function representing the rewards $R=\{throughput, \;coverage, \;load\;balance\}$ and penalties $P=\{energy, \;collision\}$. We have normalized each objective against its theoretical capacity or statistical maximum, transforming all components into dimensionless utility ratios $\in [0,1]$~\cite{yan2010survey}.
We now define all the rewards (R) and penalties (P) where user is denoted by $u$ and UAV is denoted by $i$.

Normalized \textit{throughput reward}:
\begin{equation}
\label{eq:thr}
\Phi_{\text{thr}}(R_{\text{thr}}) = \frac{\sum_{(i,u)} B \log_2(1 + \text{SNR}_{i,u})}{N_{\text{users}} \cdot B \cdot \log_2(1 + \text{SNR}_{\max})},
\end{equation}
where $\text{SNR}_{\max} = P_{tx} \cdot g(h_{\min})/N_0$ is the theoretical maximum at minimum altitude.

Normalized \textit{coverage reward}:
\begin{equation}
    \Phi_{\text{cov}}(R_{\text{cov}}) = \frac{\sum_u w_u \cdot Q_u}{\sum_u w_u},
\end{equation}
where $Q_u$ is a binary coverage quality variable indicating accurate failure probability estimation. Computation of $Q_u$ is shown in Algorithm~\ref{alg:pis1}.

% This converts coverage into a weighted percentage independent of user count.

Normalized \textit{load balance} reward:
\begin{equation}
\Phi_{\text{bal}}(R_{\text{bal}}) = \frac{\bar{n} - \sigma_n}{\bar{n} + \epsilon},
\end{equation}
where $\bar{n}$ is mean user count per UAV and $\sigma_n$ is standard deviation. 

Normalized \textit{energy consumption} penalty:
\begin{equation}
\Phi_{\text{energy}}(P_{\text{energy}}) = 1 - \frac{\sum_i \|\Delta \mathbf{p}_i\|}{N_{\text{UAV}} \cdot v_{\max} \cdot \Delta t},
\end{equation}

Normalized \textit{collision} penalty:
\begin{equation}
\label{eq:coll}
\Phi_{\text{coll}}(P_{\text{coll}}) = \frac{1}{N_{\text{pairs}}} \sum_{i<j} \exp\left(-\frac{|d_{ij} - d_{shift}|}{k_{\text{scale}}}\right),
\end{equation}
where $N_{pairs}$ is the number of UAV pairs, $d_{ij}$ is the distance between the UAVs $i$ and $j$ and $d_{shift}$ and $k_{scale}$ denote the appropriate shifting and scaling parameters.

Rather than manually tuning the weights $\omega_i$, we employ a learnable weight network $\pi_\omega: \mathcal{S} \to \Delta^4$, where $\Delta^4$ is 5-dimensional probability simplex that dynamically adjusts priorities of the reward components based on global state. 

\textbf{Algorithm \ref{alg:repositioning}} discusses the integration of above computed reward function and uses centralized training decentralized execution (CTDE) for coordinated trajectory planning. At the start of an episode, the weight network $\pi_\omega$ analyzes the global state to adaptively set objective priorities $\omega$ (Line 3). During execution, Line 8 calls Algorithm~\ref{alg:pis1} to compute $Q_u$ for every user. This value feeds directly into the normalized coverage reward $\Phi_{cov}$ (Line 11), which is aggregated with throughput, load balance, and safety penalties into the global reward $R_t$. Coordination is enforced during the training phase (Line 15), where the centralized critic $Q_i$ updates based on the \textit{joint} actions of all UAVs, effectively learning to penalize non-cooperative behaviors like collisions or redundant coverage. 
The computational complexity of Algorithm \ref{alg:repositioning} is analyzed below. Let $J = |\mathcal{I}|$ denote the number of UAV agents, $K_{\text{URLLC}}$ denote the number of URLLC users requiring PIS verification, and $N$ denote the number of importance samples per user in Algorithm \ref{alg:pis1}. The actor $\mu$ and critic $Q$ networks have $L_{\mu}$ and $L_Q$ layers respectively, where the $l$-th layer of the actor network has $U_l$ neurons and the $m$-th layer of the critic network has $V_m$ neurons. The weight network $\pi_\omega$ has $|\pi_\omega|$ parameters. Therefore, the training complexity is $\mathcal{O}(M \cdot T \cdot (J^2|\mathcal{B}|\sum_{l=1}^{L_{\mu}} U_{l-1}U_l + J|\mathcal{B}|\sum_{m=1}^{L_Q} V_{m-1}V_m + J \cdot K_{\text{URLLC}} \cdot N))$, where $|\mathcal{B}|$ is the mini-batch size. The deployment complexity is $\mathcal{O}(T \cdot (J\sum_{l=1}^{L_{\mu}} U_{l-1}U_l + J \cdot K_{\text{URLLC}} \cdot N)) = \mathcal{O}(T(J|s||a| + J \cdot K_{\text{URLLC}} \cdot N))$, where $|s|$ and $|a|$ are the dimensions of state and action, respectively.

\begin{algorithm}[h!]
\caption{PIS-Guided Multi-UAV Trajectory planning}
\label{alg:repositioning}
\begin{algorithmic}[1]
\REQUIRE Agents $\mathcal{I}$, Env $\mathcal{E}$, WeightNet $\pi_\omega$, Actor $\mu$, Critic $Q$.
\FOR{episode $ep = 1$ to $M$}
    \STATE Observe joint state $\mathbf{s} = \{s_1,..., s_N\}$
    \STATE Get adaptive weights: $\mathbf{\omega} \leftarrow \pi_\omega(\mathbf{s}_{global})$
    \FOR{step $t = 1$ to $T$}
        \STATE Select actions $\mathbf{a}_t = \{\mu_i(s_i) + \mathcal{N}_t\}_{i \in \mathcal{I}}$
        \STATE Execute actions: $\mathbf{s}' \leftarrow \text{Step}(\mathbf{a}_t)$
        \FOR{each UAV $i \in \mathcal{I}$}
            \STATE \textbf{Call Alg. 1} for assigned users to get $Q_u$
        \ENDFOR
        \STATE Calculate Normalized Components [Eq. ~\eqref{eq:thr}-~\eqref{eq:coll}]:
        \STATE \quad $\Phi_{thr}$, $\Phi_{cov}$ (via $Q_u$), $\Phi_{bal}$,$\Phi_{eng}$, $\Phi_{col}$.
        \STATE Compute Global Reward: $R_t \leftarrow \sum \omega_j \Phi_j + \sum \omega_k \Phi_k$
        \STATE Store transition $\langle \mathbf{s}, \mathbf{a}, R_t, \mathbf{s}' \rangle$ in Replay Buffer $\mathcal{D}$
        % \STATE \textbf{Multi-Agent Training (Centralized Critic):}
        \STATE Sample batch $\mathcal{B}$ from $\mathcal{D}$
        % \STATE Update Critics $Q_i$ minimizing $L(Q_i, R_t + \gamma Q'_i)$
        \STATE \textit{Coordination:} Update Critic $Q_i$ using Joint Actions $\mathbf{a}_{1..M}$ and State $\mathbf{S}_{global}$
        \STATE Update Actors $\mu_i$ via deterministic policy gradient
        \STATE Update WeightNet $\pi_\omega$ to max long-term utility
        \STATE Update states: $\mathbf{s} \leftarrow \mathbf{s}'$
    \ENDFOR
\ENDFOR
\end{algorithmic}
\end{algorithm}

\begin{figure*}[htb]
    \centering
    % First Subfigure
    
    % Second Subfigure
    \subfloat[]{
        \includegraphics[width=0.32\textwidth]{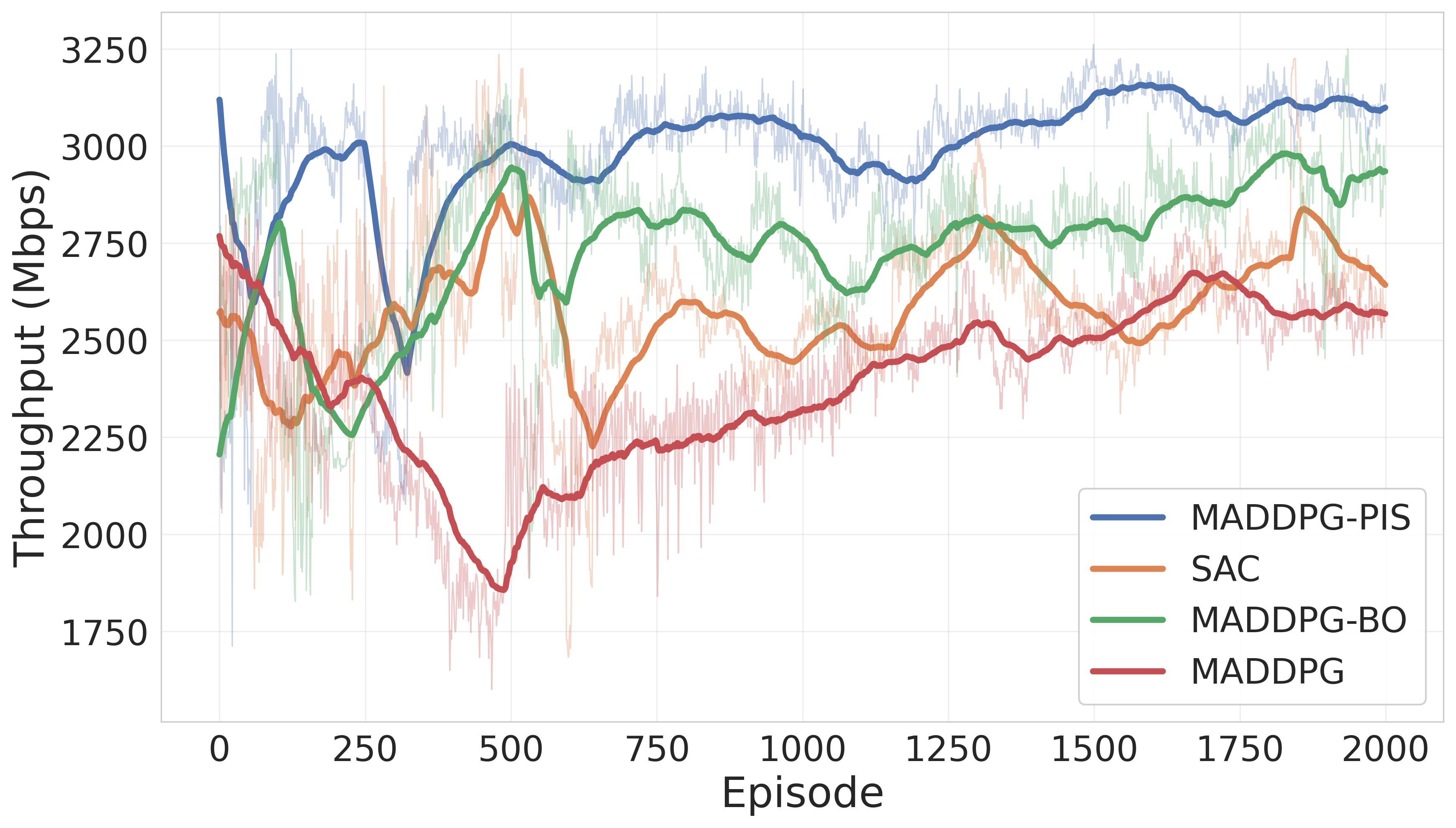}
        \label{fig:learning}
    }
    \hfil % Adds flexible space between images
    \subfloat[]{
        \includegraphics[width=0.32\textwidth]{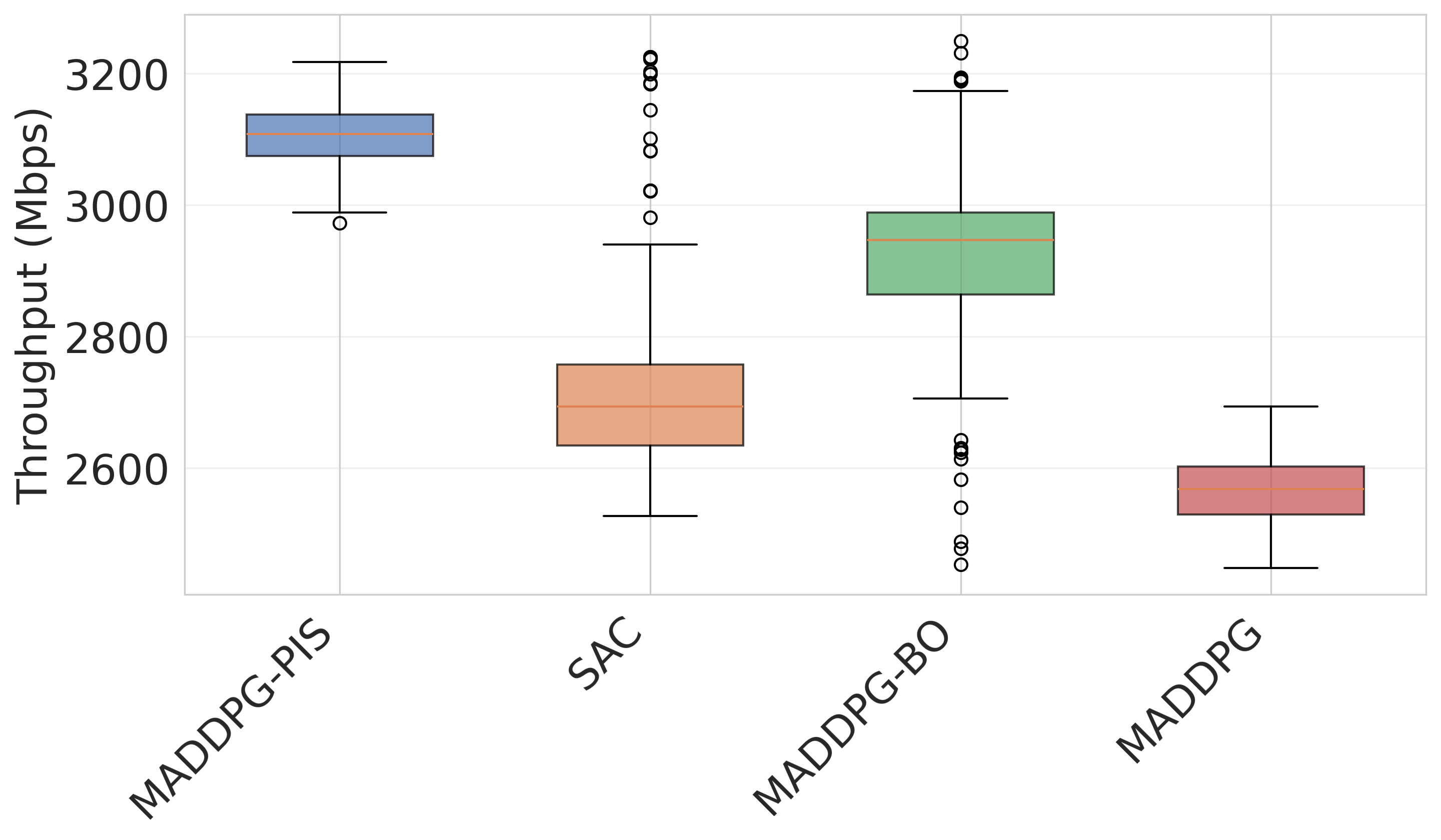}
        \label{fig:boxplot}
    }
    \hfil
    % Third Subfigure
    \subfloat[]{
        \includegraphics[width=0.29\textwidth]{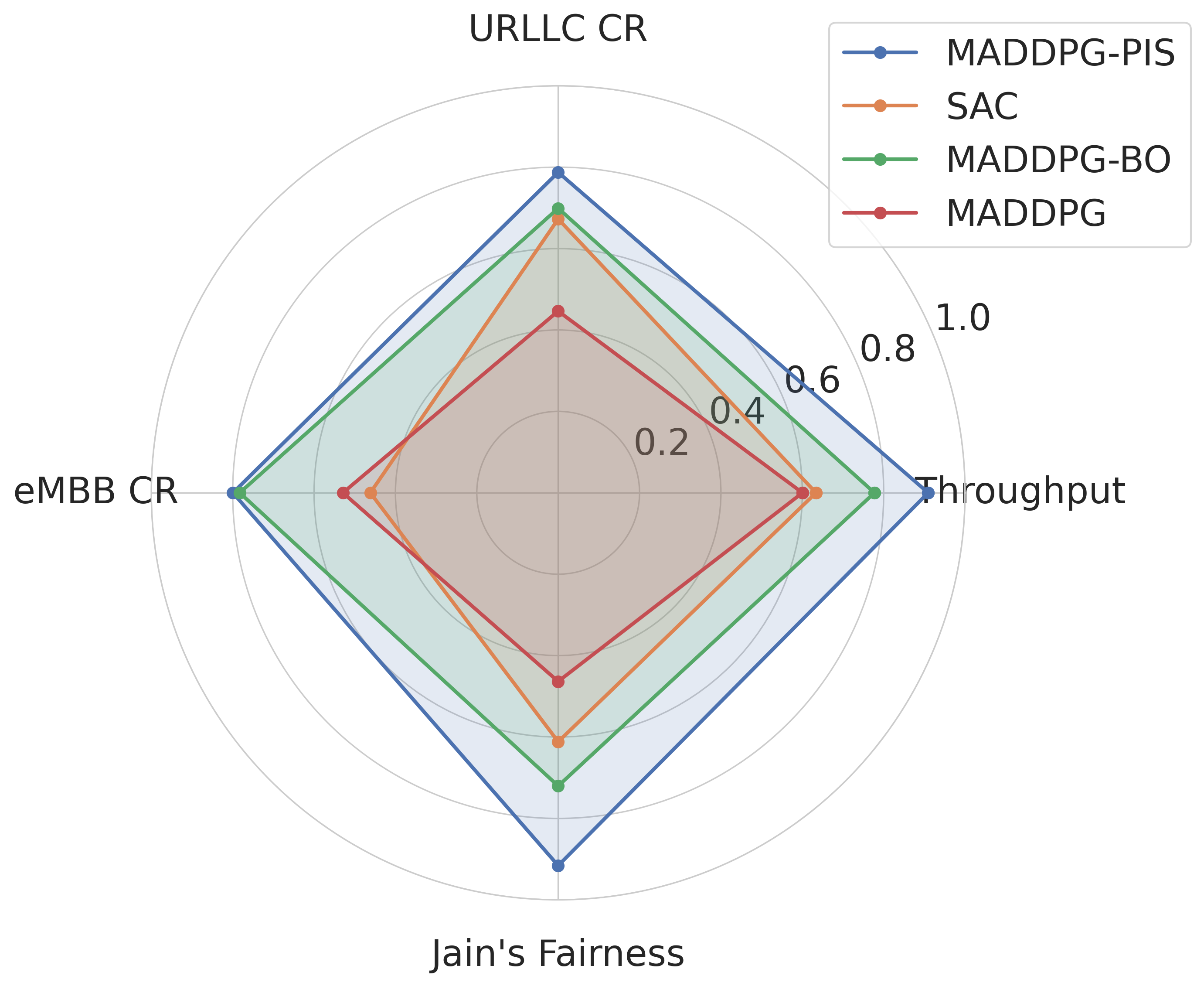}
        \label{fig:radar_chart}
    }
    
    \caption{(a) Episode-wise throughput(Mbps) (b) Throughput(Mbps) distribution (final 200 episodes) (c) Normalized performance comparison.}
    \label{fig:three_plots}
\end{figure*}
\vspace{-0.2cm}
\section{Simulation Results}
\label{simul_res}
% \subsection{Baselines:}
\noindent We compare our proposed scheme (MADDPG-PIS) with three other schemes. MADDPG~\cite{VanillaMADDPG2025}, which is a multi-agent actor-critic framework addressing environmental non-stationarity by incorporating the policies of other agents into a centralized critic. SAC~\cite{noma_2025}, which is an off-policy soft actor-critic algorithm designed to maximize computation throughput  jointly optimizing UAV trajectories and beamforming while maintaining URLLC awareness and training stability. MADDPG-BO~\cite{MADDPG_BO2025}, which is an enhanced multi-agent approach that integrates Bayesian optimization for efficient hyperparameter tuning to optimize UAV trajectories, trying to improve coverage in rural and urban environments. Moreover, we consider uniform-1000 which is a uniform sampling with $N=1000$ as accuracy baseline.  The simulation parameters we have used in the process are listed in Table \ref{tab:simulation_parameters}.
\begin{table}[h]
    \centering
    \caption{Simulation Parameters}
    \begin{tabular}{|c|c|}
        \hline
        \textbf{Parameter} & \textbf{Value} \\
        \hline
        Simulation area ($a\times a$ $m^2$) & $1500m \times 1500m$ \\
        Number of users ($N_{\text{users}}$) & 30\;(URLLC), 200\;(eMBB)\\
        Number of UAVs ($N_{\text{UAV}}$) & 5 \\
        UAV communication range & 300 m \\
        % Coverage Radius ($R$) & $46m$ \\
        Decision interval ($\Delta t$) & $10$ sec \\
        Transmission power ($P_i$) & $30$ dBm ~\cite{Heterogenous2025} \\
        UAV maximum speed ($v_{max}$) & 30 m/s     ~\cite{MADDPG_BO2025} \\
        % User maximum speed ($u_{uav}$) &   \\
        Mixture weight ($\alpha$) &   0.6      \\
        LSTM hidden dimension($h_t$) & 128 \\
        Buffer size & 50000 \\
        Rician fading factor & $2$ ~\cite{noma_2025} \\
        Altitude of UAV & $[22 \text{ m}, 150 \text{ m}]$ ~\cite{3D_Analysis2022} \\
        Carrier frequency & $73$ GHz \cite{Heterogenous2025} \\
        LoS link parameters & $\alpha_{pl} = 69.8$, $\beta_{pl} = 2$, $\sigma_0 = 3.1$ \cite{Heterogenous2025} \\
        % NLoS link parameters & $\alpha = 82.7$, $\beta = 2.69$, $\sigma_0 = 8.7$ \cite{ref16} \\
        
        \hline
    \end{tabular}
    \label{tab:simulation_parameters}
\end{table}
\begin{table}[t]
\centering
\caption{Performance Comparison}
\label{tab:results}
% \resizebox{\columnwidth}{!}{...} scales the table to the text width
\resizebox{\columnwidth}{!}{%
\begin{tabular}{lccccc}
\toprule
Method & Throughput & URLLC & eMBB & Energy & Jain's \\
 & (Mbps) & CR(\%) & CR(\%) & (kJ) & index \\
\midrule
Proposed method & \textbf{3127.85} & \textbf{74.1} & \textbf{61.7} & 116.1 & \textbf{0.71} \\
MADDPG+BO & 2906.49 & 64 & 59.8 & 123.6 & 0.61 \\
SAC & 2724.62 & 63.8 & 51.5 & 130.3 & 0.59 \\
MADDPG & 2563.03 & 55.2 & 51.5 & 92.1 & 0.52 \\
\bottomrule
\end{tabular}%
}
\end{table}

Fig.~\ref{fig:three_plots} and Table~\ref{tab:results} present a comprehensive evaluation of the proposed MADDPG-PIS framework against baseline methods across convergence, stability, and multi-objective performance dimensions.
The episode-wise throughput trajectories in Fig.~\ref{fig:learning}) reveal learning dynamics over 2000 training episodes. Raw measurements (transparent lines) exhibit inherent stochastic variance, while the smoothed curves highlight convergence trends. MADDPG-PIS demonstrates superior asymptotic performance, converging to $\sim$3100 Mbps by episode 1500. Notably, MADDPG exhibits prolonged initial instability (episodes 0-500) before gradual convergence, whereas PIS-guided exploration accelerates early-stage learning due to optimistic initialization. 
Box plots (Fig.~\ref{fig:boxplot}) characterize steady-state performance distributions over the final 200 episodes. The interquartile range (IQR) quantifies policy stability where MADDPG-PIS exhibits the tightest distribution (IQR = 62.89 Mbps), indicating consistent performance with minimal variance. In contrast, SAC, MADDPG-BO, MADDPG has IQR (Mbps), 123.43, 124.57 and 73.76 respectively. The median line position within each box confirms that MADDPG-PIS maintains higher central tendency without extreme deviations.
The radar chart (Fig.~\ref{fig:radar_chart}) normalizes four critical metrics to [0,1] for comparative visualization. MADDPG-PIS achieves the largest enclosed area, demonstrating balanced excellence across multiple competing objectives. Specifically, it dominates in throughput, URLLC coverage ratio  and Jain's Fairness index, while maintaining eMBB coverage. The rhomboid shape of baseline MADDPG reveals its bias toward lower fairness at the expense of reliability. MADDPG-BO shows improved coverage over SAC but lags in fairness index.
Absolute performance metrics over converged episodes (Table~\ref{tab:results}) confirm MADDPG-PIS has a upper hand: 3127.85 Mbps throughput (+22\% vs. vanilla MADDPG, +14.8\% vs. SAC and +7.6\% vs. MADDPG-BO), 74.1\% URLLC CR (+10.3\% vs. MADDPG-BO), and 0.71 Jain's index (+16.4\% vs. MADDPG-BO). 
% \subsection{Ablation Study}
% Table~\ref{tab:ablation} shows systematic parameter variation in our ablation study. The proposed model ($\alpha=0.6, N=100$) achieves optimal performance: 3097 Mbps throughput, 72\% URLLC coverage, and 0.57 ms latency per user. Increasing $\alpha$ to 0.9 (aggressive reliance on prediction) degrades performance to 2219 Mbps and 50.6\% CR due to insufficient uniform fallback exposing prediction error risk. Conversely, reducing $\alpha$ to 0.3 (conservative mixture) yields robust but still inefficient performance (2376 Mbps, 52.8\% CR) by under-utilizing predictive information. Uniform sampling ($N=1000$) achieves comparable performance (2998 Mbps, 68.1\% CR) but requires 11 ms latency which is 20 $\times$ slower than PIS, making it impractical for real-time control. This ablation directly confirms PIS's core value proposition of achieving equivalent accuracy with 20 $\times$ speedup through intelligent importance sampling.

Table~\ref{tab:ablation} presents our ablation study. The proposed ($\alpha=0.6, N=100$) achieves optimal performance (3097 Mbps, 72\% URLLC CR, 0.57 ms latency) with estimator variance $5.2 \times 10^{-5}$, marginally higher than Uniform ($N=1000$) at $4.8 \times 10^{-5}$. This is expected since Theorem \ref{thm:variance} guarantees variance reduction at equal $N$, not across different sample budgets. Notably, the $10\times$ sample reduction causes only a marginal increase rather than proportional increase that naive uniform subsampling would incur, empirically demonstrating PIS's superior per-sample efficiency. High $\alpha=0.9$ degrades throughput to 2219 Mbps and 50.6\% CR as insufficient uniform fallback exposes prediction errors, with inflated variance ($29.8 \times 10^{-5}$) reflecting unbounded importance weights under MDN misprediction. Low $\alpha=0.3$ under-utilizes prediction, yielding 2376 Mbps and 52.8\% CR with variance ($19.1 \times 10^{-5}$) as the uniform fallback dominates. Uniform ($N=1000$) achieves comparable accuracy but at 11 ms latency which is $20\times$ slower than PIS thereby rendering it impractical for real-time control. 
% This ablation directly confirms PIS's core value proposition of achieving equivalent accuracy with $20\times$ speedup through intelligent importance sampling.

% Table~\ref{tab:ablation} identifies the contribution of PIS components by varying parameters in a systematic way. In the rare event regime ($P_f$ = $4.7\times10^{-5}$), standard uniform sampling with N=2000 samples produce $11$ ms latency averaged over all users. Our PIS method achieves comparable accuracy (86.42\%) using only N=100 samples in 20ms, demonstrating 20$\times$ computational speedup while maintaining estimation quality for real-time URLLC coverage verification.

% In the rare event regime ($P_f$ = $4.7\times10^{-5}$), standard uniform sampling requires N=2000 samples to achieve 81.54\% relative error with 391ms latency averaged over all users. Our PIS method achieves comparable accuracy (86.42\%) using only N=100 samples in 20ms, demonstrating 20$\times$ computational speedup while maintaining estimation quality for real-time URLLC coverage verification.

The convergence of evidence across visualization modalities establishes that PIS-guided coverage verification enables MADDPG to simultaneously optimize multiple conflicting objectives while maintaining policy stability, confirming theoretical variance reduction guarantees derived in Theorems (\ref{thm:unbiased} \& \ref{thm:variance}).
\vspace{-0.3cm}
% \begin{table}[h!]
% \centering
% \caption{Effect of PIS Parameters}
% \label{tab:ablation}
% \begin{tabular}{lccccc}
% \toprule
% Configuration & Throughput & URLLC & Latency per \\
%  & (Mbps) & CR(\%) & user(ms) \\
% \midrule
% Proposed ($\alpha=0.6$) & \textbf{3097} & \textbf{72} & 0.57\\
% High $\alpha$ ($\alpha=0.9$) & 2219 & 50.6 & 0.42\\
% Low $\alpha$ ($\alpha=0.3$) & 2376 & 52.8 & 0.64\\
% % Uniform ($N=50$) & 2345 & 48.3 & 20 \\
% Uniform ($N=1000$) & 2998 & 68.1 & 11\\
% \bottomrule
% \end{tabular}
% \end{table}

\begin{table}[htbp]
\centering
\caption{Effect of PIS Parameters}
\label{tab:ablation}
% Uncomment the \resizebox line below if the table is still wider than your column margin
\resizebox{\columnwidth}{!}{%
\begin{tabular}{lcccc}
\toprule
Configuration & \makecell{Throughput \\ (Mbps)} & \makecell{URLLC \\ CR(\%)} & \makecell{$\text{Var}(\hat{P}_f)$ \\ $(\times 10^{-5})$} & \makecell{Latency per \\ user(ms)} \\
\midrule
Proposed ($\alpha=0.6$) & \textbf{3097} & \textbf{72} & 5.2 & 0.57 \\
High $\alpha$ ($\alpha=0.9$) & 2219 & 50.6 & 29.8 & 0.42 \\
Low $\alpha$ ($\alpha=0.3$) & 2376 & 52.8 & 19.1 & 0.64 \\
Uniform ($N=1000$) & 2998 & 68.1 & 4.8 & 11 \\
\bottomrule
\end{tabular}
} % Uncomment this closing brace if you use \resizebox
\end{table}

% \begin{table}[h!]
% \centering
% \caption{Effect of PIS Parameters}
% \label{tab:ablation}
% \begin{tabular}{lcccc}
% \toprule
% Configuration & Throughput & URLLC & Time per & Relative\\
%  & (Mbps) & CR(\%) & user(ms) & error(\%)\eqref{rel_err}\\
% \midrule
% Proposed ($\alpha=0.6, N=100$) & \textbf{3097} & \textbf{72} & 0.087 &  86.42\\
% High $\alpha$ ($\alpha=0.9$) & 2376 & 52.6 & 0.072 & 219.73\\
% Low $\alpha$ ($\alpha=0.3$) & 2219 & 50.8 & 0.087 & 108.73\\
% % Uniform ($N=50$) & 2345 & 48.3 & 20 \\
% Uniform ($N=1000$) & 2998 & 68.1 & 2 & 81.54\\
% \bottomrule
% \end{tabular}
% \end{table}
% \subsection{The Greedy dilema}

% While Greedy achieves competitive throughput (2699.78 Mbps), this metric is misleading due to Shannon capacity's logarithmic relationship with SNR. A UAV hovering directly over a single user achieves higher aggregate throughput than serving multiple users from a central location (the "centroid trap"). Greedy exploits this by serving only nearest users with high-quality links while abandoning distant ones, resulting in high peak rates but poor fairness and lower coverage (37.76\% vs PIS's 60.52\%). In contrast, PIS-MADDPG optimizes for fair coverage across all users—a critical URLLC requirement that aggregate metrics fail to capture. This analysis underscores the importance of evaluating both throughput and coverage rate jointly, rather than relying on throughput alone as a performance indicator.
\vspace{-0.2cm}
\section{Conclusion}
\label{conclusion}
\noindent We presented a PIS based theoretically-grounded framework for sample-efficient coverage verification in multi-UAV networks. By combining LSTM-MDN trajectory prediction with defensive mixture importance sampling, we achieve provable variance reduction while maintaining unbiasedness. Integration with MADDPG enables real-time proactive control for multi-UAV trajectory planning. Experimental validation demonstrates the superiority of our proposed approach over existing approaches like vanilla MADDPG, SAC, MADDPG-BO and pure uniform sampling. Future work includes adaptive mixture control (learning $\alpha$ online based on prediction uncertainty), extension to 3D mobility for aerial users, and hardware validation on real UAV testbeds. Moreover, the proposed PIS framework can be generalized beyond UAV networks to any mobile wireless system requiring sample-efficient spatial coverage verification.

\end{document}